\documentclass[leqno]{article}
\usepackage{graphicx}
\usepackage{indentfirst,csquotes}

\usepackage[utf8]{inputenc}
\usepackage{xcolor}
\usepackage{bm}
\title
{Superintegrability and Deformed Oscillator Realizations of Quantum TTW Hamiltonians on Constant-Curvature Manifolds and with Reflections in a Plane}
\author{Ian Marquette$^{1}$ and Anthony Parr$^{2}$}

\date{$^{1}$ Department of Mathematical and Physical Sciences, La Trobe University \\ Bendigo 3552, Victoria Australia \\ $^{2}$ School of Mathematics and Physics, The University of Queensland \\ Brisbane, QLD 4072, Australia}

\usepackage{amsmath,amssymb,amsfonts,amsthm,mathtools}
\mathtoolsset{showonlyrefs}

\newcommand{\mrm}[1]{\mathrm{#1}}
\newcommand{\mbb}[1]{\mathbb{#1}}
\newcommand{\mcl}[1]{\mathcal{#1}}
\newcommand{\mfk}[1]{\mathfrak{#1}}
\DeclareMathOperator{\sgn}{sgn}

\DeclareMathOperator*{\Res}{Res}
\theoremstyle{theorem}
\newtheorem{theorem}{Theorem}[section]
\newtheorem{proposition}[theorem]{Proposition}
\newtheorem{corollary}[theorem]{Corollary}
\newtheorem{lemma}[theorem]{Lemma}

\usepackage{mathrsfs}  

\numberwithin{equation} {section}
\usepackage[maxbibnames=99,backend=bibtex,sorting=none]{biblatex}
\bibliography{References.bib}
\begin{document}
\maketitle
\begin{abstract}
We extend the method for constructing symmetry operators of higher order for two-dimensional quantum Hamiltonians by Kalnins, Kress and Miller (2010). This expansion method expresses the integral in a finite power series in terms of lower degree integrals so as to exhibit it as a first-order differential operators. One advantage of this approach is that it does not require the \textit{a priori} knowledge of the explicit eigenfunctions of the Hamiltonian nor the action of their raising and lowering operators as in their recurrence approach (2011). We obtain insight into the two-dimensional Hamiltonians of radial oscillator type with general second-order differential operators for the angular variable. We then re-examine the Hamiltonian of Tremblay, Turbiner and Winternitz (2009) as well as a deformation discovered by Post, Vinet and Zhedanov (2011) which possesses reflection operators. We will extend the analysis to spaces of constant curvature. We present explicit formulas for the integrals and the symmetry algebra, the Casimir invariant and oscillator realizations with finite-dimensional irreps which fill a gap in the literature.
\end{abstract}

\\
email: i.marquette@latrobe.edu.au, anthony.parr@uq.net.au
\bigskip
\noindent

\tableofcontents

\section{Introduction}
Let  \(H\) be a quantum Hamiltonian on an \(N\)-dimensional Riemannian manifold:
\[H=-\tfrac{1}{2}\Delta+V(\mathbf{x})\]
\sloppy where $\Delta$ is the Laplace-Beltrami operator. An integral (of \(H\)) is a differential operator with meromorphic coefficients which commutes with \(H\). A family of integrals \(I_1,I_2,\ldots\) is independent if there is a Jordan polynomial identity \(P(I_1,I_2,\ldots)=0\) that resolves into a sum \(\sum_i Q_i(I_i)+\sum_{i<j}Q_{ij}(I_i,I_j)+\cdots\) of symmetric polynomials. The Hamiltonian is a superintegrable system  if there are more than \(N\) independent integrals (including the Hamiltonian) in which \(N\) of them pairwise commute. The system is of order \(d\) if the maximal order of the integrals, as differential operators, is \(d\). In classical mechanics, the number of independent integrals is at most \(2N-1\). This is typically the case for quantum systems also, and a model which reaches this limit in number of integrals is said to be maximally superintegrable. This family of independent integrals forms a generating set for a subspace of all the integrals of the Hamiltonian, called the symmetry algebra. The symmetry algebra contains elements of arbitrarily large order but is finitely generated. These algebras have steadily acquired their own status among the different algebraic structures of interest for application in mathematical physics or other areas, such as finite and infinite-dimensional Lie algebras and super Lie algebras. Repeated commutators of the generating integrals may produce additional integrals which are not independent. Thus, the symmetry algebra is a polynomial Lie algebra under the commutator operation \([\cdot,\cdot]\) with additional relations. These polynomials are of Jordan-type and cross-multiplication can be concisely expressed with the anticommutator \(\{\cdot,\cdot\}\). 

The exceptional properties of superintegrable systems have excited considerable interest in recent years, particularly their connection with separability 
and the algebraic approach to solving the Schr\"{o}dinger equation. 
The classification of second-order superintegrable systems on conformally flat two-dimensional space is well-understood \cite{mak67,win67,ev90, das01,post11a,mil13a, mil13b}, and current attention is directed towards finding and solving models in higher dimensions \cite{kal1, das10,es17, cap15, hoque15a, hoque15b,liao18,da19}  and order \cite{gra,mar08,tre10, post15,esc18a,esc18b, mar17, abou, mar19, rita,ritb}.

Some systems of third-, fourth- and fifth-order have been studied \cite{gra,tre10,esc18a, mar17}. Searching for potentials such that resulting Hamiltonian is superintegrable involves non-linear compatibility equations, so that higher order systems are difficult to study. Indeed, solutions for these modest orders yield potentials of Painlev\'{e}, Chazy class and even beyond. 

Several construction procedures have been developed for finding higher order integrals for a potential specified up to a few rational or integral parameters \cite{post15, esc18a, esc18b, miller10,bern20a, bern20b, cal4a, mar15, kal11, mar14}. Beyond a direct expansion in powers of the momenta \cite{post15,esc18a,esc18b}, we have for separable systems that the symmetry algebra corresponds to the dynamic algebra of a simpler Hamiltonian, usually taking the form of generalized Heisenberg algebras \cite{cal4a,mar15} but obtaining explicit forms is a non-trivial problem. In the quantum case, these algebras are built up from ladder operators and the recurrence relations \cite{kal11} of orthogonal polynomials and special functions from the Askey-Wilson scheme. 

The expansion of integrals as a finite power series of the Hamiltonian and other known integrals so as to represent the whole as a differential operator of reduced order has proved effective \cite{miller10,bern20a, bern20b}. This technique, which we shall designate the expansion method, has been applied to the resolution of third-order integrals on Riemannian manifolds and to obtain various exotic potentials which depend upon Painlev\'{e} transcendents. The expansion method is powerful in that no knowledge of the eigenfunctions are required. Its application to the resolution of higher order integrals \cite{miller10} was met with difficulties in producing explicit results, because the equations involved were of fourth degree.

It is our purpose here to apply the expansion and ladder operator methods to the construction of the symmetry algebra of certain separable Hamiltonians. We shall, throughout, work over a two-dimensional space of constant curvature \(\kappa\). Positive curvature corresponds to elliptical space, negative curvature to hyperbolic space and no curvature to flat Euclidean space. The Hamiltonian families of interest here are of radial harmonic oscillator type
\begin{equation}
H=\frac{1}{2}\bigg[-\partial_r^2-\frac{c}{s}\partial_r+\frac{M^2}{s^2}+\left(\omega^2-\tfrac{1}{4}\kappa^2\right)\frac{ s^2}{c^2}\bigg]\label{eq:hamiltonian}
\end{equation}
where: \(\omega>0\); \(M^2\) is a second-order differential operator which depends on the angular variable \(\theta\) and a rational parameter \(k\in\mbb{Q}_{>0}\); 
\begin{align*}
s&=\sum^\infty_{i=0}\frac{(-1)^i\kappa^ir^{2i+1}}{(2i+1)!}=\begin{cases}\frac{\sin(\sqrt{\kappa}r)}{\sqrt{\kappa}},&\quad \kappa >0,\\r,&\quad \kappa=0,\\\frac{\sinh(\sqrt{-\kappa}r)}{\sqrt{-\kappa}},&\quad \kappa <0,\end{cases}
\end{align*}
and 
\begin{align*}
c&=\sum^\infty_{i=0}\frac{(-1)^i\kappa^ir^{2i}}{(2i)!}=\begin{cases}\cos(\sqrt{\kappa}r),&\quad \kappa >0,\\1,&\quad \kappa =0,\\\cosh(\sqrt{-\kappa}r),&\quad \kappa < 0,\end{cases}
\end{align*}
are the trigonometric, flat or hyperbolic sine and cosine.The expansion method, as applied to this Hamiltonian, consists in writing
\[\partial_r^n=\sum_{\ell,m}[F_{0\ell mn}(s,c)+F_{1\ell mn}(s,c)\partial_r](M^2)^\ell H^m.\]
We may thus consolidate any differential operator in the variables \(r,\theta\) into the form
\begin{equation}
I=I_0(r,\theta,\partial_\theta,H)+I_1\bm{(}r,\theta,\partial_\theta,\rho(H)\bm{)}\partial_r\label{eq:Iab}
\end{equation}
where \(I_0\) and \(I_1\) are meromorphic in the first two arguments and polynomial in the last two, arranged so that the \(H\)-terms are ordered to the right of the variables and operators. Here, \(\rho\) denotes multiplication from the right, i.e. \(\rho(H)\partial_r=\partial_rH\). This can be extended to a large class of superintegrable systems.

The paper is organized in the following way. The main results is that we determine the full symmetry algebra which is not normally achieved through a straightforward application of the ladder operator method (see, e.g., Theorem \ref{kred}). In \S 2, we introduce the two classes of models which we are of focus for this paper, the TTW and PVZ models, and discuss the different forms that their symmetries will take. The approach we shall take in finding their integrals \(I\) is to expand as a power series in \(s\). In \S 3, we find that an infinite expansion corresponds to an infinite-order differential operator. The finite expansion enforces upon the operator the constraints that it is a linear expansion of formal operators \(L_i\) satisfying
\begin{equation}
[M^2,L_i]=4piL_i(pi+M)\label{eq:ladder}
\end{equation}
where \(k=p/q\) in reduced form. The ladder operators \(L_i\) decompose into a product of an angular operator \(J\) and radial operator \(K\), which we determine by similar expansion techniques in \S 4. In \S 5, the generators of the symmetry algebra are presented, their commutation and closure relations are given and finite-dimensional irreducible representations are derived. The latter will give us the spectrum of the Hamiltonian, obtained through purely algebraic means.
\section{Models, Eigenfunctions and Spectra}
In this section, we discuss the TTW and PVZ models which are of focus for this paper and give their spectral decomposition, which will be used for comparison with the representation theory of their symmetry algebras to be developed in \S 5. We shall discuss the differences between how their symmetry algebras will be constructed. 

The celebrated TTW model \cite{ttw09,ttw10} has the Posch-Teller operator as angular component:
\begin{equation}
M^2_{\text{TTW}}=-\partial_\theta^2+\frac{k^2}{4}\left[\frac{4\alpha^2-1}{\sin^2(k\theta)}+\frac{4\beta^2-1}{\cos^2(k\theta)}\right],\label{eq:m2ttw}
\end{equation}
where \(\alpha,\beta >-1\). The Hamiltonian \eqref{eq:hamiltonian} is separable in polar coordinates. Therefore the eigenfunctions can be determined analytically by simultaneously diagonalizing \(H\) and \(M^2\). The eigenfunctions of \eqref{eq:m2ttw} are
\[\Theta_\ell(\theta)=[\sin(k\theta)]^{\alpha+\frac{1}{2}}[\cos(k\theta)]^{\beta+\frac{1}{2}}\mathcal{P}_\ell^{\alpha,\beta}\bm{(}\cos(2k\theta)\bm{)},\qquad \ell=0,1,2,\ldots\]
where \(\mathcal{P}\) is a Jacobi polynomial. The corresponding eigenvalues are \(k^2(2\ell+\alpha+\beta+1)^2\). 

This model has attracted intensive interest \cite{que10a,que10b,cal12,gon12,cel12,
cha14,ran14,hak17,ran17} as a higher-order superintegrable system distinct from the anisotropic oscillator. The TTW model includes, for \(k=1,2,3\), the Smorodinsky-Winternitz model \cite{smod65}, the \(BC_2\) system of Type V in the classification of Olshanetsky and Perelomov \cite{per}, and the three-body model introduced by Wolfes \cite{wolfes} respectively. The TTW systems in their full generality were introduced in 2009 by Tremblay, Turbiner and Winternitz where they verified the model was superintegrable for \(k=4\) and conjectured superintegrability for all rational \(k\). 
Kalnins, Kress and Miller Jr. \cite{kal11} and Quesne \cite{que10b} proposed various approaches to demonstrating superintegrability. Studies of this model have generated new problem-solving techniques \cite{miller10, kal11} which have been applied to other models that are structurally very similar to the TTW model \cite{post12a}. 
The ladder operator method of Kalnins \textit{et al.} involved decomposing the integral into formal components \(L_{\pm 1}\) of the form \(\tfrac{1}{2}I_{\mathrm{I}}\pm \tfrac{1}{2}I_{\mathrm{II}}M\), \(I_{\mathrm{I}},I_{\mathrm{II}}\) being differential operators. Then \(I_{\mathrm{I}}=L_{1}+I_{-1}\), \(I_{\mathrm{II}}=(L_1-L_{-1})/M\) are integrals of \(H\). Even for \(k=1\), this does not yield the lowest degree operators. Indeed, there is a differential operator \(I_{\mathrm{III}}\) and a polynomial \(P\) such that
\begin{equation}
   I_{\mathrm{I}}-pI_{\mathrm{II}}=I_{\mathrm{III}}(M^2-p^2)+P(H).\label{eq:d3}
\end{equation}
The existence of the integral \(I_{\mathrm{III}}\) was inferred by Kalnins \textit{et al.} from algebraic considerations and computed an expression for \(P(H)\). That \(M^2,H,I_{\mathrm{III}},H\) generate by multiplication the full symmetry algebra for all \(k\) is not evident from their analysis. That this is the case we demonstrate in \S 5.3. Moreover, we derive equation \eqref{eq:d3} from the explicit differential operator form. We also provide an algebraic derivation of the spectrum and a generalization of the results for curved spaces. The classical analogue \cite{ttw10} is well-understood and indeed has been shown to be superintegrable on curved space \cite{ran14,ran11}.

We shall also consider a model discovered by Post, Vinet and Zhedanov \cite{post11b} (PVZ) which has angular momentum component
\begin{equation}
M^2_{\text{PVZ}}=-\partial_\theta^2+k^2\left[\frac{\alpha^2-2\alpha\cos(2k\theta)R}{\sin^2(2k\theta)}+\frac{\beta^2-2\beta\sin2(k\theta)}{\cos^2(2k\theta)}\right],\label{eq:m2_pvz}
\end{equation}
where \(\alpha,\beta >-1\). Here, \(R\) is a reflection operator on functions of \(\theta\), i.e. \((Rf)(\theta)=f(-\theta)\). The PVZ model is similar to the TTW model but exhibits its own peculiar properties. In particular, \(M^2\) is a perfect square:
\begin{equation}
M_{\text{PVZ}}=\left[\partial_\theta-\frac{k\beta}{\cos(2k\theta)}\right]R-\frac{k\alpha}{\sin(2k\theta)}.\label{eq:mpvz}
\end{equation}
The eigenfunctions of \eqref{eq:m2_pvz} are \[\Theta_\ell(\theta)=[\sin(2k\theta)]^{\frac{\alpha}{2}}[\cos(2k\theta)]^{\frac{\beta}{2}}[1+\sin(2k\theta)]^{\frac{1}{2}}\mathscr{P}_\ell\bm{(}\sin(2k\theta);\alpha,\beta|-1\bm{)},\qquad \ell=0,1,2,\ldots\]
where \(\mathscr{P}\) is a Little Jacobi polynomial \cite{little}. The eigenvalues of \eqref{eq:mpvz} are alternating in sign and are given by
\[(-1)^{\ell+1}k(2\ell+\alpha+\beta+1).\]
A symmetry algebra had been obtained for this system utilizing the ladder operator method. The ladder operators satisfy generalized Heisenberg algebra relations:
\begin{align*}
ML_1&=(-1)^qL_1(M+2p),\\
ML_{-1}&=(-1)^qL_1(M-2p),\\
[L_1,L_{-1}]&=P(M,H)
\end{align*}
for some polynomial \(P\). The algebra stated by Post \textit{et al.} was generated by \(M,H,L_1,L_{-1}\) for \(q\) even and \(M,H,L_2,L_{-2}\) for \(q\) odd. However, for \(q\) odd, we have the property that
\[L_1=\tilde{I}_\mathrm{I}(M+p)+\xi_+,\qquad L_{-1}=\tilde{I}_\mathrm{II}(M-p)+\xi_-\]
where \(\tilde{I}_\mathrm{I},\tilde{I}_\mathrm{II}\) are differential-difference operators and \(\xi_\pm\) are constants. Then the full symmetry algebra is given by \(M,H,\tilde{I}_\mathrm{I},\tilde{I}_{\mathrm{II}}\), which we show in detail in \S 5.2. The model for even \(q\) does not possess this property, and forms only a simple polynomial algebra satisfying the generalized Heisenberg relations.

The reduction properties associated with the TTW model and the PVZ model for \(q\) even, do not appear to necessarily be specific to these models alone but a consequence of the commutation identities associated with the ladder operators, which are common to a variety of systems. 

The eigenfunctions for both systems can be expressed as \(\Psi_{m,\ell}(r,\theta)=S_{m,\ell}(r)\Theta_\ell(\theta)\) where the radial part is found to be
\[S_{m,\ell}(r)=\begin{cases}s^{k(2\ell+\alpha+\beta+1)}c^{\frac{1}{2}+\frac{\omega}{|\kappa|}}\mathcal{P}^{k(2\ell+\alpha+\beta+1),\frac{\omega}{|\kappa|}}_m(c^2-\kappa s^2),&\quad \kappa\neq 0\\
r^{k(2\ell+\alpha+\beta+1)}\mathcal{L}^{k(2\ell+\alpha+\beta+1)}_m(\omega r^2),&\quad \kappa =0 
\end{cases}\qquad m=0,1,2,\ldots\]
and \(\mathcal{L}\) is a Laguerre polynomial. The spectrum is
\begin{equation}
E_{m,\ell}=\omega\epsilon_{m,\ell}+\tfrac{1}{2}\kappa \epsilon_{m,\ell}^2,\qquad \epsilon_{m,\ell}=2m+1+k(2\ell+\alpha+\beta+1).\label{eq:spectrum}
\end{equation}
The spectrum of \(H\) therefore splits into finite degeneracy levels \(\{E_{m,\ell}\}_{qm+p\ell=n}\). That is, the eigenvalues are invariant under the transformation \((m,\ell)\mapsto (m\pm p,\ell\mp q)\) which is the implicit evidence for integral operators being built up out of ladder operators, that raise or lower the quantum numbers of the eigenfunctions.
\section{General Form of the Symmetry}\label{sec3}
In this section, we demonstrate from the expansion method that every integral of the TTW and PVZ systems must be a linear combination of a product of ladder operators \(J^uK^{ku}\) where \(J^u=J^u(\theta,\partial_\theta,R,M)\) and \(K^{ku}=K^{ku}(r,\partial_r,M,H)\) which satisfy the commutation relations:
\begin{align}
\label{eq:jkrule}[M^2,J^u]&=4kuJ^u(ku+M),&[K^u,H]&=\frac{2u(u+M)}{s^2}K^u.
\end{align}
In the case of the TTW model, each equation is really two, with \(M\) being a formal operator. 
%

For an operator \(I\) consolidated in the manner of \eqref{eq:Iab}, we write 
\(I_{,r},I_{,\theta},\ldots\) as its partial derivatives with respect to just the variables \(r,\theta\) ignoring any of the operator terms. Then \([I,H]=0\) generates the following system:
\begin{subequations}
\begin{align}
\frac{[M^2,I_0]}{s^2}&=I_{0,r,r}+\frac{c}{s}I_{0,r}+\frac{2}{s^2}\rho_{M^2}I_{1,r}\label{eq:conabe1}-4I_{1,r}H\\
&\qquad+\frac{2s}{c^3}\left(\omega^2-\tfrac{1}{4}\kappa^2\right)(scI_{1,r}+I_1)-\frac{2c}{s^3}I_1M^2,\\
\frac{[M^2,I_1]}{s^2}&=I_{1,r,r}-\frac{c}{s}I_{1,r}+\frac{1}{s^2}I_1\label{eq:conabe2}+2I_{0,r},  
\end{align}
\end{subequations}
For \(\kappa=0\), taking a power series expansion in \(r\) to solve the system \eqref{eq:conabe1}--\eqref{eq:conabe2} naturally suggests itself. For general \(\kappa\), we find that the appropriate \textit{Ansatz} is:
\begin{align*}
I_0&=\sum_is^{-i}I_{0,i}(\theta,\partial_\theta,R,H),& I_1&=\sum_ics^{-i+1}I_{1,i}(\theta,\partial_\theta,R,H).
\end{align*}
We obtain the recurrence equations:
\begin{subequations}
\begin{align}
0&=\kappa(i+1)(i+2)I_{0,i+2}-[i^2-M^2+\rho(M^2)]I_{0,i}\label{eq:i0}\\
&\qquad+2(i+2)\left(\omega^2-\tfrac{1}{4}\kappa^2+2\kappa H\right)I_{1,i+4}\\
&\qquad-2(i+1)[2H+\kappa \rho(M^2)]I_{1,i+2}+2iI_{1,i}M^2,\\
0&=\kappa i(i+1)I_{1,i+2}-[i^2- M^2+\rho(M^2)]I_{1,i}+2iI_{0,i}.\label{eq:i1}
\end{align}
\end{subequations}
Combining \eqref{eq:i0} with \eqref{eq:i1} in order to eliminate the \((I_{0,i})\)-sequence, we get
\begin{align}
0&=i(i+2)[4\omega^2+8\kappa H-\kappa^2(i+2)^2]I_{1,i+4}\label{eq:aeq}\\
&\qquad+2i(i+1)\{\kappa[(i+1)^2+1-M^2-\rho(M^2)]-4H\}I_{1,i+2}\nonumber\\
&\qquad+\{4i^2\rho(M^2)-[i^2-M^2+\rho(M^2)]^2\}I_{1,i}
\end{align}
It is not possible for this sequence to admit infinitely many non-zero values without \(I\) having infinite degree in \(H,M^2\). Therefore, we must truncate the series.

If \(I_{1,m}\neq 0\) but \(I_{1,i}=0\) for \(i<m\), then equation \eqref{eq:aeq} says
\[[4\omega^2+8\kappa H-\kappa^2(m-2)^2] (m-2)(m-4)I_{1,m}=0,\]
so \(m\) is \(2\) or \(4\). Furthermore, \(I_{0,i}=I_{1,i}=0\) for odd \(i\).

In order to cut off the upper bound on the coefficients, we require that for some sufficiently large integer \(\mu\)
\[\{[\mu^2-M^2+\rho(M^2)]^2-4\mu^2\rho(M^2)\}I_{1,\mu}=0.\]
Equation \eqref{eq:aeq} gives
\[\prod^{\frac{1}{2}\mu}_{i=\frac{1}{2}j}\{[4i^2-M^2+\rho(M^2)]^2-16i^2\rho(M^2)]I_{1,j}=0\]
and, from \eqref{eq:i1}, we have
\[\prod_{i=1}^{\frac{1}{2}\mu}\{[4i^2-M^2+\rho(M^2)]^2-16i^2\rho(M^2)\}[M^2,I]=0\]
with \(\mu\) even. From here we wish to show \(I\) is a linear combination of operators \(L_i\) satisfying \eqref{eq:ladder} where \(i\) is any integer with magnitude less than or equal to \(\tfrac{1}{2}\mu\). Let \(Q_{i^2}\) satisfy
\begin{equation}
\bm{[}M^2,[M^2,Q_{i^2}]\bm{]}-8i^2\{M^2,Q_{i^2}\}+16i^4Q_{i^2}=0.\label{eq:comi}
\end{equation}
for a non-zero integer \(i\). Then
\[L_i\coloneqq \frac{1}{4i}[M^2,Q_{i^2}]-iQ_{i^2}+Q_{i^2}M\]
satisfies \eqref{eq:ladder}. Conversely, any \(L_i\) satisfying this commutation relation also satisfies \eqref{eq:comi} in place of \(Q_{i^2}\). Now, if
\[\bm{[}M^2,[M^2,L]\bm{]}-8i^2\{M^2,L\}+16i^2L=L_j\]
where \(i\neq j\), then \(L\) resolves into homogeneous solutions \(L_{\pm i}\) plus the particular solution
\[L_j[16(i^2-j^2)(i^2-j^2-2jM-M^2)]^{-1}.\]
This achieves the desired decomposition. It is evident that the \(\{L_i\}_i\) are linearly independent as operators with coefficients formally rational in \(M,H\). It follows that each \(L_i\) must commute with \(H\) also. Furthermore, there is a finite basis of operators \(L_i\) satisfying \eqref{eq:ladder} and we may choose a basis \(J^{i/p}(\theta,\partial_\theta,R,M)\) consisting of operators with no dependence upon \(r,\partial_r\). Then \(L_i=J^{i/p}K^i\) where \(K^i=K^i(r,\partial_r,M,H)\). These must satisfy \eqref{eq:jkrule}. The basis \(J^{i/p}\) can be further decomposed into operators satisfying either one of the equations
\begin{equation}
[M,J^u]=2ku J^u\qquad\text{or}\qquad \{M,J^u\}=-2kuJ^u.\label{eq:pvzJ}
\end{equation}
However, for the TTW model, this is merely an ambiguity on the sign of \(M\) acting on the eigenfunctions of \(M\). We take its action to always be of the same sign, i.e., we formally set \(MJ^u=J^u(M+2ku)\). The operators \(J,K\) modify respectively the quantum numbers \(\ell,m\) of the eigenfunctions of \(H\).
\section{Ladder Operators}
In the previous section we established that the integral must be built up from linear combinations of ladder operators. We shall in this section derive the ladder operators for the values of the exponents in which they exist. This is the integers. We consider the angular ladder operators \(J\) for the PVZ and TTW model and finally, the radial ladder operator \(K\) which is the same for both. The result of this section is that the form of the integral must be
\[I=\sum_iJ^{iq}K^{ip}\chi_i(M,H)\]
where we permit \(\chi_i\) to be rational functions insofar as \(I\) is a differential or differential-difference operator. Having calculated the ladder operators, we determine what forms of coefficients are acceptable, that being when their fractional residue in \(M,H\) (treated as formal operators) on the whole vanishes. The TTW model has the additional constraint that \(I\) must be even in \(M\).

We use the expansion method to simplify the angular partial derivatives. For the TTW model, we write
\[\partial_\theta^m=\sum_\ell[F_{2\ell m}\bm{(}\sin(k\theta),\cos(k\theta)\bm{)}+F_{3\ell m}\bm{(}\sin(k\theta),\cos(k\theta)\bm{)}\partial_\theta](M^2_{\text{TTW}})^\ell\]
in order to realize \(J^u\) as a first-order operator in the form
\begin{equation}
J^u(\theta,\partial_\theta,M)=J_{0}^u(\theta)(M)+J_1^u(\theta)(\rho_M) \partial_\theta.\label{eq:jnotettw}
\end{equation}
For the PVZ model, we do the same by writing
\[\partial_\theta^m=\sum_\ell[F_{4\ell m}\bm{(}\sin(2k\theta),\cos(2k\theta)\bm{)}+F_{5\ell m}\bm{(}\sin(2k\theta),\cos(2k\theta)\bm{)}R]M^\ell_{\text{PVZ}}\]
Hence, \(J^u\) can be written as
\begin{equation}
J^u(\theta,\partial_\theta,R)=J_0^u(\theta)(M)+J_1^u(\theta)(\rho_M)R.\label{eq:jnotepvz}
\end{equation}
We can employ an analogous reduction for \(K^u\):
\begin{equation}
K^u(r,\partial_r,M)=K^u_0(r,M,H)+K^u_1(r,M,H)\partial_r.\label{eq:knote}
\end{equation}
In \cite{kal11}, the reducibility of the integrals was induced from simple cases (i.e. \(k=1\)), but there was no way to verify that the results obtained were in fact of minimal degree. This is then an important problem: to develop systematically a means to obtain the lowest-order integrals. Here we will develop the method for the TTW and PVZ models on constant-curvature space based on the explicit results of \S\S 4.1--3.

For polynomials in \(M,H\), we consider their divisibility where \(M,H\) are taken to be formal operators. The greatest common divisors of a family of polynomials \(P=\{p_i(M,H)\}_i\) are the polynomials \(Q\) which divide \(P\) such that all the divisors of \(P\) are divisors of \(Q\). For a given representative \(q\in Q\), we shall write \(\gcd P\propto q\). The units are the non-zero real numbers. A formal differential operator of the form \(\sum_i\sum_{0\leq a,b\leq 1} f_{abi}(r,\theta)\partial_r^a\partial_\theta^bg_{abi}(M,H)\) where \(g_{abi}\) are polynomials has a gcd given by \(\gcd\{g_{abi}\}\). For the PVZ model, we substitute \(R\) for \(\partial_\theta\).
\subsection{PVZ Angular Operators}
We determine the coefficients in \eqref{eq:jnotepvz} by finding solutions to \eqref{eq:pvzJ} which are polynomials in \(M\). The resultant system is
\begin{subequations}
\begin{align}
{J_0^u}'(-\theta)&=\frac{k\alpha}{\sin(2k\theta)}[J_1^u(-\theta)-J_1^u(\theta)]\label{eq:pvzeq1}\\&\qquad-J_1^u(-\theta)M-vJ_1^u(\theta)(M+2ku);\\
{J_1^u}'(-\theta)&=\frac{k\alpha}{\sin(2k\theta)}[J_0^u(-\theta)-J_0^u(\theta)]-\frac{2k\beta}{\cos(2k\theta)}J_1^u(-\theta)\label{eq:pvzeq2}\\
&\qquad+J^u_0(-\theta) M-vJ^u_0(\theta)(M+2ku),
\end{align}
\end{subequations}
where \(v=\pm 1\) is to be determined.
To solve this differential-difference system, we develop the coefficients into a sinusoidal power series with a cosine gauge on \(J^u_1\):
\[J^u_0(\theta)(M)=\sum_{i=V}^U[\sin(2k\theta)]^iJ^u_{0,i}(M),\qquad J^u_1(\theta)(M)=\cos(2k\theta)\sum^{\widetilde{U}}_{i=\widetilde{V}}[\sin(2k\theta)]^iJ^u_{1,i}(M).\]
Substituting these series into \eqref{eq:pvzeq1}--\eqref{eq:pvzeq2}, we get the recurrence relations
\begin{subequations}
\begin{align}
0&=\{[(-1)^i-v]M/k-2uv\}J^u_{0,i}-[(-1)^i+1]\alpha J^u_{0,i+1}\label{eq:j1}
-2(-1)^i(i+1)J^u_{1,i+1}\\
&\qquad-2(-1)^i\beta J^u_{1,i}+2(-1)^iiJ^u_{1,i-1};\\
0&=2(-1)^i(i+1)J^u_{0,i+1}+\{[(-1)^i+v]M/k+2uv\}J^u_{1,i}\label{eq:j2}+[(-1)^i+1]\alpha J^u_{1,i+1}.
\end{align}
\end{subequations}
For finite-order in \(M\), the system \eqref{eq:j1}--\eqref{eq:j2} requires \(J^u_{0,i}=J^u_{1,i}=0\) for all but finite \(i\). The appropriate limits are \(U=|u|\), \(\widetilde{U}=|u|-1\), \(V=\widetilde{V}=0\). In particular,
\begin{align*}
0&=\{[(-1)^u-v]M/k-2uv\}J^u_{0,|u|}+2(-1)^u|u|J^u_{1,|u|-1},\\
0&=2(-1)^u|u|J^u_{0,|u|}+\{[(-1)^u-v]M/k-2uv\}J^u_{1,|u|-1},
\end{align*}
which, for \([J^u_{0,|u|}(M)]^2+[J^u_{1,|u|-1}(M)]^2\neq 0\), implies
\[\{[(-1)^u-v]M/k-2uv\}^2=4u^2\]
so \(v=(-1)^u\). Thus, \(\{M,J^u\}=-2kuJ^u\) when \(u\) is odd and \([M,J^u]=2kuJ^u\) when \(u\) is even. 
\noindent Up to proportionality the solution is unique, so we designate our principle solutions to be:
\begin{align}
J^{\pm 1}&=[\sin(2k\theta)\pm \cos(2k\theta)R]( M\pm k)+k(\alpha \pm \beta)\label{eq:j05}
\end{align}
and define the rest by \(J^{\pm 2j}=(J^{\mp 1}J^{\pm 1})^j\) and \(J^{\pm 2j\pm 1}=J^{\pm 1}J^{\pm 2j}=J^{\mp 2j}J^{\pm 1}\). In particular,
\begin{align}
J^{\pm 2}&=[\cos(4k\theta)\mp\sin(4k\theta)R](M\pm k)(M\pm 3k)\label{eq:jone}\\
&\qquad \mp k\sin(2k\theta)[k\alpha+\beta(M\pm 2k)]\pm k\cos(2k\theta)R[k\beta+\alpha(M\pm 2k)]\\
&\qquad+k^2(\alpha-\beta)(\alpha+\beta);
\end{align}
Note that in \eqref{eq:j05} but not in \eqref{eq:jone}, all the coefficients except the constant term are divisible by a linear factor in \(M\). This, in fact, occurs for every odd \(u\).
\begin{lemma}
For all odd integers \(u\), \(J^u(\theta)(-ku)=J^u_{0,0}(-ku)\).\label{half-int}
\end{lemma}
\begin{proof}
This is clearly true for \(|u|=1\) and we assume the lemma holds for \(|u|=2j-1\), where \(j\) is a positive integer. Using \(J^{\pm 2j\pm 1}=J^{\pm 1}J^{\mp 2j\pm 1}J^{\pm 1}\), we have
\begin{align}
J^{\pm 2j\pm 1}(\theta)\bm{(}\mp (2j+1)k\bm{)}
&=J^{\mp 2j\pm 1}_{0,0}\bm{(}\pm (2j-1)k\bm{)}\Phi^{\pm 1}\bm{(}\mp(2j+1)k\bm{)}\label{eq:lem41}
\end{align}
which is a constant.
\end{proof}
The combinations \((J^u)^2\) and \(J^{-u}J^u\) for \(u\) odd and even respectively, commute with \(M\). We designate them both by \(\Phi^u(M)\), and is appropriately called the structure function. By induction, we compute:
\begin{align}
\Phi^{\pm j}(M)&=\prod^{\lceil j/2\rceil-1}_{i=0}[k^2(\alpha\pm\beta)^2-(M\pm k\pm 4ki)^2]\prod^{\lfloor j/2\rfloor-1}_{i=0}[k^2(\alpha\mp \beta)^2-(M\pm 3k\pm 4ki)^2].\label{eq:oddj}
\end{align}
For \(u\) even, we have \(\Phi^{-u}(M)=\Phi^u(M-2ku)\) and \(\Phi^u(M)=\Phi^j(-M-2ku)\) for \(u\) odd. With knowledge of the structure function, we may now use equation \eqref{eq:lem41} to compute \(J^{\pm (2j+1)}_{0,0}\bm{(}\mp k(2j+1)\bm{)}\). The result of the computation, consolidated into a single formula, is:
\[J^u_{0,0}(-ku)=k^{|u|}\left[\alpha+(-1)^{\frac{u-1}{2}}\beta\right]\prod^{\frac{|u|-1}{2}}_{i=1}\left\{\left[\alpha+(-1)^{\frac{u-1}{2}-i}\beta\right]^2-4i^2\right\}\]
where \(u\) is any odd integer. By the lemma, it follows that
\begin{equation}
    [J^u-J^u_{0,0}(-ku)](M+ku)^{-1}\label{eq:jhalfdivisor}
    \end{equation}
is polynomial in \(M\). 
We now need to prove that there are no redundant factors in \(J^u\) generally. First, we compute some explicit terms in \(J^u\) so as to illustrate the phenomenon which we shall exploit. For a concise expression, we use
\[\exp(\theta R)=\cos(\theta)+\sin(\theta) R.\]
Then for \(j>1\):
\begin{align}
J^{\pm 2j\mp 1}&=\pm R\mrm{e}^{\mp 2(2j-1)k\theta R}\prod^{2j-1}_{i=1}[M\pm (2i-1)k]\label{eq:jhlf}\\
&\qquad+k(\alpha\pm\beta)\mrm{e}^{\pm 4(j-1)k\theta R}\prod^{2j-1}_{i=2}[M\pm (2i-1)k]\\
&\qquad+k(\alpha\pm\beta)\mrm{e}^{\mp 4(j-1)k\theta R}\prod^{2(j-1)}_{i=1}[M\pm (2i-1)k]\\&\qquad\pm k^2(\alpha\pm\beta)^2R\mrm{e}^{\pm 2(2j-3)k\theta R}\prod^{2(j-1)}_{i=2}[M\pm (2i-1)k]\\
&\qquad\pm 2(j-1)k^2(\alpha^2-\beta^2)R\mrm{e}^{\mp 2(2j-3)k\theta R}\prod^{2(j-1)}_{i=2}[M\pm (2i-1)k]+\cdots
\end{align}
\begin{align}
J^{\pm 2j}&=\mrm{e}^{\mp 4jk\theta R}\prod^{2j}_{i=1}[M\pm (2i-1)k]\label{eq:jful}\\&\qquad\pm k(\alpha\pm\beta)R\mrm{e}^{\pm 2(2j-1)k\theta R}\prod^{2j}_{i=2}[M\pm (2i-1)k]\\
&\qquad\pm k(\alpha\mp\beta)R\mrm{e}^{\mp 2(2j-1)k\theta R}\prod^{2j-1}_{i=1}[M\pm (2i-1)k]\\&\qquad+k^2(\alpha^2-\beta^2)\mrm{e}^{\pm 4(j-1)k\theta R}\prod^{2j-1}_{i=2}[M\pm (2i-1)k]\\
&\qquad+k^2(2j-1)(\alpha^2-\beta^2)\mrm{e}^{\mp 4(j-1)k\theta R}\prod^{2j-1}_{i=2}[M\pm (2i-1)k]+\cdots
\end{align}
where the ellipses indicates trigonometric functions with lower frequency. We observe in the above that the first few terms in \eqref{eq:jhlf} and \eqref{eq:jful} consist of products of \(M+(2i-1)k\) where the number of factors is decreasing along with the frequency of the \(\mrm{e}^{k\ell\theta R}\) terms which they multiply. In particular, the common divisor of the second and third terms is the same as that of the fourth and fifth terms. This is why \(J^u\) for \(u\) even cannot be reduced after removing the constant term, but \(J^u\) for \(u\) odd can be.
\begin{theorem}
Suppose \(\alpha+\beta,\alpha-\beta\) are not even integers. Then for all odd integer \(u\), \[\gcd\{J^u_{0,i},J^u_{1,i}\}_{i\geq 0}\propto 1\]
and for all non-zero even integer \(u\),
\[\gcd\{J^u_{0,i+1},J^u_{1,i}\}_{i\geq 0}\propto 1\]\label{jpvzred}
\end{theorem}
\begin{proof}
From \eqref{eq:jhlf} and \eqref{eq:jful},
\[\gcd\{J^u_{0,|u|}(M),J^u_{0,|u|-1}(M)\}\propto\prod^{|u|}_{i=1}[M+\sgn(u)(2i-1)k]\]
so we need to show that there are non-zero terms in \(J^u\) when \(M\) is evaluated at each of the roots. By writing
\[J^u=J^{-u+\sgn(u)(2j+1)}J^{\sgn(u)(2j+1)},\qquad 0\leq j<|u|,\]
we have
\[J^u_{0,|u|-2j-1}\bm{(}-\sgn(u)(2j+1)k\bm{)}\neq 0\]
since \(J^{\sgn(u)(2j+1)}_{0,0}\bm{(}-\sgn(u)(2j+1)\bm{)}\neq 0\) by the hypothesis. Similarly, by writing
\[J^{u-\sgn(u)(2j+1)}J^{2u-\sgn(u)(2j+1)},\qquad 0\leq j<|u|\]
we get
\[J^u_{0,|u|-2j-1}\bm{(}-2uk+\sgn(u)(2j+1)k\bm{)}\neq 0.\]
The result follows.
\end{proof}
Theorem \ref{jpvzred} will be combined with similar results for the \(K\) operators to determine that only the factor \(M+ku\) will divide the non-constant terms of \(J^uK^{ku}\) for \(u\) odd but this is indivisible for \(u\) even.
\subsection{TTW Angular Operators}
\label{ttwangularoperatorsection}
Here we will solve equation \eqref{eq:jkrule} for \(M^2\) given by equation \eqref{eq:m2ttw}. Substituting \eqref{eq:jnotettw} into \eqref{eq:jkrule} produces the system
\begin{align*}
0&={J^u_0}''+u(u+\rho_M)J^u_0-2{J^u_1}'M^2+\frac{k^2(4\alpha^2-1)}{2\sin^2(k\theta)}[{J^u_1}'-k\cot(k\theta)J^u_1]\\
&\qquad+\frac{k^2(4\beta^2-1)}{2\cos^2(k\theta)}[{J^u_1}'+k\tan(k\theta)J^u_1],\\
0&={J^u_1}''+u(u+\rho_M)J^u_1+2{J^u_0}'.
\end{align*}
We develop the coefficients as a cosinusoidal power series
\begin{equation}
J^u_0(\theta)(M)=\sum^{U}_{i=V}[\cos(2k\theta)]^iJ^u_{0,i}(M),\qquad J^u_1=\sin(2k\theta)\sum^{\widetilde{U}}_{i=\widetilde{V}}[\cos(2k\theta)]^iJ^u_{1,i}(M)
\end{equation}
to obtain the recurrence relations:
\begin{subequations}
\begin{align}
0&=k(i+1)(i+2)J^u_{0,i+2}+[u(ku+M)-ki^2]J^u_{0,i}e\label{eq:ttwjeq2}\\&\qquad+(i+1)[M^2-k^2(2\alpha^2+2\beta^2-1)]J^u_{1,i+1}\nonumber\\
&\qquad+k^2(\beta^2-\alpha^2)(2i+1)J^u_{1,i}-M^2iJ^u_{1,i-1}\\
0&=k(i+1)(i+2)J^u_{1,i+2}+[u(ku+M)-k(i+1)^2]J^u_{1,i}\label{eq:ttwjeq1}\\&\qquad-(i+1)J^u_{0,i+1}.
\end{align}
\end{subequations}
Again, we require \(U,\widetilde{U},V,\widetilde{V}\) to be finite. The necessary limits are \(U=|u|\), \(\widetilde{U}=|u|-1\), \(V=\widetilde{V}=0\). We take as our principles
\[J^{\pm 1}=\sin(2k\theta)\partial_\theta(k\pm M)\pm \cos(2k\theta)M(k\pm M)+k^2(\alpha^2-\beta^2)\]
and set \(J^{\pm j}=(J^{\pm 1})^j\). From \(J^{-1}=J[1;-M]\) we get \(J^{-u}=J[u;-M]\). In particular, we obtain a differential operator with no formal part if we consider symmetric combinations:
\begin{equation}
J^uf(M)+J^{-u}f(-M)\label{eq:symcomb}
\end{equation}
where \(f\) is polynomial. We will need to determine, as with the PVZ model, to what extent \(f\) may be rational. Indeed, setting \(f(M)=M^{-1}\) in \eqref{eq:symcomb} leaves no residue.

For \(j>1\), we have:
\begin{align}
J^j&=[\sin(2jk\theta)\partial_\theta+\cos(2jk\theta)M]\prod^{2j-1}_{i=1}(M+ki) \label{eq:jgen}\\&\qquad+2k^2(\alpha^2-\beta^2)[(j-1)\sin[2(j-1)k\theta]\partial_\theta\\
&\qquad+\cos[2(j-1)k\theta] (jM+2jk-k)]\prod^{2j-2}_{i=2}(M+ki)+\cdots
\end{align}
where the ellipses denote trigonometric terms of lower frequency. The structure function \(\Phi^u(M)\) is defined as \(J^{-u}J^u\). For positive values, we determine:
\begin{align*}
\Phi^j(M)&=\prod^{j-1}_{i=0}[(M+ 2ki+ k)^2-k^2(\alpha+\beta)^2][(M+2ki+ k)^2-k^2(\alpha-\beta)^2].
\end{align*}
Observe \(\Phi^{-j}(M)=\Phi^j(-M)=\Phi^j(M-2jk)\).

The divisibility property which held for odd integers in \S 4.1 this time holds for all integers.
\begin{lemma}
For all non-zero integer \(u\), \(J^u(\theta)(-ku)=J^u_{0,0}(-ku)\).\label{thjttw}
\end{lemma}
\begin{proof}
We know that \(J^u_{0,|u|}(-ku)=J^u_{1,|u|-1}(-ku)=0\). Now, suppose that \(J^u_{0,j}(-ku)=J^u_{1,j-1}(-ku)=0\) for all \(j>i\) where \(0<i<|u|\). Equations \eqref{eq:ttwjeq1} and \eqref{eq:ttwjeq2} simplify to \begin{align*}
0&=-ki^2J^u_{0,i}(-ku)-k^2u^2iJ^u_{0,i-1}(-ku);\\
0&=-ki^2J^u_{1,i-1}(-ku)-iJ^u_{0,i}(-ku).\\
\end{align*}
This system is non-singular so \(J^u_{0,i}(-ku)=J^u_{1,i-1}(-ku)=0\) and therefore the result holds by induction.
\end{proof}
\noindent Consequently,
\[[J^{\pm u}-J^{\pm u}_{0,0}(-ku)](M\pm ku)^{-1}\]
is of the form \(\pm D_0+D_1M\) where \(D_0,D_1\) are differential operators. 
We wish to find an explicit formula for \(J^u_{0,0}(-ku)\). To this end, we express the structure function in terms of the coefficients:
\begin{align*}
\Phi^j(M)&=J^j_{0,0}(-M-2jk)J^j_{0,0}(M) -2kJ^j_{1,0}(-M-2jk)J^j_{0,1}(M)\\
&\qquad+[k^2(2\alpha^2+2\beta^2-1)-M^2]J^j_{1,0}(-M-2jk)J^j_{1,0}(M).
\end{align*}
By Lemma \ref{thjttw}
\begin{equation}\Phi^j(-jk)=[J^j_{0,0}(-jk)]^2\label{eq:struc}
\end{equation}
This gives \(J^j_{0,0}(-jk)\) up to a sign. By \(J^j=J^1J^{j-1}\), we have:
\begin{align*}
J^j_{0,0}(M)&=J^1_{0,0}\bm{(}M+2k(j-1)\bm{)}J^{j-1}_{0,0}(M)-2kJ^1_{1,0}\bm{(}M+2k(j-1)\bm{)}J^{j-1}_{0,1}(M)\\
&\qquad+[k^2(2\alpha^2+2\beta^2-1)-M^2]J^1_{1,0}\bm{(}M+2k(j-1)\bm{)}J^{j-1}_{1,0}(M),\\
J^j_{0,1}(M)&=J^1_{0,1}\bm{(}M+2k(j-1)\bm{)}J^{j-1}_{0,0}(M)+J^1_{0,0}\bm{(}M+2k(j-1)\bm{)}J^{j-1}_{0,1}(M)\\
&\qquad+2k^2(\alpha^2-\beta^2)J^1_{1,0}\bm{(}M+2k(j-1)\bm{)}J^{j-1}_{1,0}(M)\\
&\qquad+[k^2(2\alpha^2+2\beta^2-1)-M^2]J^1_{1,0}\bm{(}M+2k(j-1)\bm{)}J^{j-1}_{1,1}(M),\\
J^j_{1,0}(M)&=J^1_{1,0}\bm{(}M+2k(j-1)\bm{)}J^{j-1}_{0,0}(M)+J^1_{0,0}\bm{(}M+2k(j-1)\bm{)}J^{j-1}_{1,0}(M)\\
&\qquad-2kJ^1_{1,0}\bm{(}M+2k(j-1)\bm{)}J^{j-1}_{1,1}(M),\\
J^j_{1,1}(M)&=J^1_{0,0}\bm{(}M+2k(j-1)\bm{)}J^{j-1}_{1,1}(M)+2kJ^1_{1,0}\bm{(}M+2k(j-1)\bm{)}J^{j-1}_{1,0}(M)\\
&\qquad+J^1_{0,1}\bm{(}M+2k(j-1)\bm{)}J^{j-1}_{1,0}(M)+J^1_{1,0}\bm{(}M+2k(j-1)\bm{)}J^{j-1}_{0,1}(M).
\end{align*}
This indicates that in powers of \(\alpha\):
\begin{align*}
J^j_{0,0}&=k^{2j}\alpha^{2j}+\mcl{O}(\alpha^{2j-2});& J^j_{0,1} & =\mcl{O}(\alpha^{2j-2}); \\ J^j_{1,0}&=\mcl{O}(\alpha^{2j-2}); & J^j_{1,1}&=\mcl{O}(\alpha^{2j-2});
\end{align*}
for all \(j\geq 1\). This enables us to choose the correct square-root in \eqref{eq:struc}. At last, we calculate
\begin{align}
J^j_{0,0}(-jk) &=k^{2j}(\alpha^2-\beta^2)^{\frac{1-(-1)^j}{2}}\prod^{\lfloor j/2\rfloor-1}_{i=0}\{[(2i-j+1)^2-\alpha^2-\beta^2]^2-4\alpha^2\beta^2\}.\label{eq:cres}
\end{align}
The result below shows that the coefficients, except in extraordinary cases, have no common factor.
\begin{theorem}
Suppose that \(\alpha+\beta,\alpha-\beta\) are not integers. Then
\(\gcd\{J^u_{0,i},J^u_{1,i}\}_{i\geq 0}\propto 1.\)\label{jttwred}
\end{theorem}
\begin{proof}
We shall assume \(u\) is positive without loss of generality. From \eqref{eq:jgen}, we have
\[\gcd\{J^u_{0,2u},J^u_{1,2u-1}\}\propto\prod^{2u-1}_{i=1}(M+ki).\]
By \(J^u=J^{u-j}J^j\) we get
\[J^u_{0,2u-2j}(-jk) \neq 0\]
for \(1\leq j\leq u\). Now, \(J^u\bm{(}-k(j+u)\bm{)}=J^{u-j}\bm{(}k(j-u)\bm{)}J^j\bm{(}-k(j+u)\bm{)}\) so
\[J^u_{0,2j}\bm{(}-k(j+u)\bm{)}\neq 0.\]
This shows that the coefficients cannot all be equal to zero simultaneously.
\end{proof}
\begin{corollary}
\label{jcor} Let \(u,v\) be integers with \(|u|\geq |v|\). Then the operator
\[[J^u-J^{|u|-|v|}_{0,0}\bm{(}k(|v|-|u|)\bm{)} J^v](M+u+v)^{-1}\]
is polynomial in \(M\).
\end{corollary}
\subsection{Radial Operators}
For the determination of the radial operators \(K^u\), we resume the procedure of \S\ref{sec3}, namely by expanding as a power series in \(s\):
\[K^u_0(r,M,H)=\sum^{|u|}_{i=0}s^{-2i}K^u_{0,i}(M,H),\qquad K^u_1(r,M,H)=\sum^{|u|-1}_{i=0}cs^{-1-2i}K^u_{1,i}(M,H).\]
Equation \eqref{eq:jkrule} translates to the system
\begin{subequations}
\begin{align}
0&=\kappa (i+\tfrac{1}{2})(i+1)K^u_{0,i+1} -[i^2-u(u+M)]K^u_{0,i}\label{eq:db1}\\
&\qquad+(i+1)(\omega^2-\tfrac{1}{4}\kappa^2+2\kappa H)K^u_{1,i+1}\\
&\qquad-(i+\tfrac{1}{2})(\kappa M^2+2H)K^u_{1,i}+iM^2K^u_{1,i-1},\\
0&=\kappa i(i+\tfrac{1}{2})K^u_{1,i}-[i^2-u(u+M)]K^u_{1,i-1}+iK^u_{0,i}\label{eq:db2}.
\end{align}
\end{subequations}
We take as our principles
\[K^{\pm 1}=\frac{(1\pm M)c}{s}\partial_r+H+ \tfrac{1}{2}\kappa M(M\pm 1)- \frac{M(M\pm 1)}{s^2}.\]
Note that \(K^{-1}(r,\partial_r,-M)=K^1(r,\partial_r,-M)\). The effect of curvature is manifested in the \(\tfrac{1}{2}\kappa M(M\pm 1)\) term, but otherwise the form of the operators is the same as that on flat space \cite{miller10}. Unlike in the computation of the eigenfunctions in \S 2, there is no need to divide into curved and non-curved space. This will serve only to increase the order of the symmetry algebra, but not its overall structure. Indeed, the calculations to compute the integrals remains the same since the \(K\)-operators only depend explicitly on integer powers of \(\kappa,c,s,M,H\), whereas the Jacobi polynomials degenerate to the Laguerre polynomials only by a limit process. We set
\[K^{\pm j}(r,\partial_r,M)=K^{\pm 1}\bm{(}r,\partial_r,M\pm 2(j-1)\bm{)}K^{\pm (j-1)}(r,\partial_r,M)\]
which gives the identity \(K^{-u}=K^u(r,\partial_r,-M)\). Generally, we have
\begin{align}
K^j&=\frac{(-2)^{j-1}}{s^{2j-1}}\left(\frac{c}{M}\partial_r-\frac{1}{s}\right)\prod^{2j-1}_{i=0}(M+i)\label{eq:kopcoeff}\\
&\qquad+(2j-1)\frac{(-2)^{j-2}}{s^{2j-3}}\left(\frac{c}{M}\partial_r-\frac{1}{s}\right)[H+\tfrac{1}{2}\kappa M(M+2j-1)]\prod^{2j-2}_{i=1}(M+i)\\
&\qquad-\frac{(-2)^{j-2}}{s^{2j-3}}\left(\frac{c}{M}\partial_r+\frac{1}{s}\right)[H+\tfrac{1}{2}\kappa M(M+1)]\prod^{2j-1}_{i=2}(M+i)+\mcl{O}\left(\frac{1}{s^{2j-4}}\right)
\end{align}
for positive integer \(j\). The operators \(K^{-u}(r,\partial_r,M+2u)K^u(r,\partial_r,M)\) commute with \(M,H\) and are thus polynomials in those operators. We designate this operator by the structure function \(\Psi^u(M,H)\). We have the formula
\begin{equation}
\Psi^j(M,H)=\prod^{j-1}_{i=0}\left\{
[H-\tfrac{1}{2}\kappa(M+2i+1)]^2-\omega^2(M+2i+1)^2\right\}.
\end{equation}
Observe that \(\Psi^{-j}(M,H)=\Psi^j(-M,H)=\Psi^j(M-2j,H)\). The base terms have the recurrence
\begin{align*}
K^j_{0,0}(M,H)&=K^1_{0,0}(M+2j-2,H)K^{j-1}_{0,0}(M,H)\\
&\qquad+(\omega^2-\tfrac{1}{4}\kappa^2+2\kappa H)K^1_{1,0}(M+2j-2,H)K^{j-1}_{1,0}(M,H);\\
K^j_{1,0}(M,H)&=K^1_{1,0}(M+2j-2,H)K^{j-1}_{0,0}(M,H)+[\kappa K^1_{1,0}(M+2j-2,H)\\
&\qquad+K^1_{0,0}(M+2j-2,H)]K^{j-1}_{1,0}(M,H)
\end{align*}
and a simple closed-form expression is readily obtainable:
\begin{align}
K^j_{1,0}(M,H)&=\frac{1}{2\sqrt{\omega^2+2\kappa H}}\label{eq:cdcoef1}\\
&\qquad\times \Bigg\{\prod^{j-1}_{i=0}[H+\tfrac{1}{2}\kappa (M+2i+1)^2+\sqrt{\omega^2+2\kappa H}(M+2i+1)]\\
&\qquad-\prod^{j-1}_{i=0}[H+\tfrac{1}{2}\kappa (M+2i+1)^2-\sqrt{\omega^2+2\kappa H}(M+2i+1)]\Bigg\};\\
K^j_{0,0}(M,H)&=-\tfrac{1}{2}\kappa K^j_{1,0}(M,H)\label{eq:cdcoef2}\\
&\qquad+\frac{1}{2}\prod^{j-1}_{i=0}[H+\tfrac{1}{2}\kappa (M+2i+1)^2+\sqrt{\omega^2+2\kappa H}(M+2i+1)]\\
&\qquad+\frac{1}{2}\prod^{j-1}_{i=0}[H+\tfrac{1}{2}\kappa (M+2i+1)^2-\sqrt{\omega^2+2\kappa H}(M+2i+1)].
\end{align}
We see that \(K^j_{1,0}(-j,H)=0\) and the latter becomes
\begin{equation}
K^j_{0,0}(-j,H)=H^{\frac{1-(-1)^j}{2}}\prod^{\lfloor j/2\rfloor-1}_{i=0}\left\{\left[H-\tfrac{1}{2}\kappa (2i-j+1)^2\right]^2-\omega^2(2i-j+1)^2\right\}.
\end{equation}
\begin{proposition}
For all non-zero integers \(u\), \(K^u(-u,H)=K^u_{0,0}(-u,H)\).\label{5.6}
\end{proposition}
\begin{proof}
By equation \eqref{eq:kopcoeff},
\[K^u_{0,|u|}(-u,H)=K^u_{1,|u|-1}(-u,H)=0.\]
Now, suppose that for some \(0<i<|u|\), \(K^u_{0,j}(-u,H)=K^u_{1,j-1}(-u,H)=0\) for all \(j>i\). Then \eqref{eq:db1} and \eqref{eq:db2} read
\begin{align*}
0&=iK^u_{0,i}(-u,H)-i^2K^u_{1,i-1}(-u,H);\\
&=iu^2K^u_{1,i-1}(-u,H)-i^2K^u_{0,i}(-u,H).
\end{align*}
This equation is non-singular, so \(K^u_{0,i}(-u,H)=K^u_{1,i-1}(-u,H)=0\). This proves the result.
\end{proof}
\begin{lemma}
Let \(u\) be a non-zero integer, then\label{lemma}
\[L_j\coloneqq\gcd\{K^u_{0,i}(M,H),K^u_{1,i-1}(M,H)\}_{i\geq |u|-j+1}\propto\prod^{2|u|-j}_{i=j}[M+\sgn(u)i],\qquad j=1,\ldots,|u|.\]
\end{lemma}
\begin{proof}
Proposition \ref{5.6} tells us that \(M+u\) divides \(L_{|u|}\). We suppose, without loss of generality, that \(u\) is positive and \(j<u\). We write \[K^u(r,\partial_r,M)=K^{u-i}(r,\partial_r,M+2i)K^i(r,\partial_r,M)\] to obtain
\begin{equation}
K^u(r,\partial_r,-j)=K^{u-j}(r,\partial_r,j)K^j_{0,0}(-j,H).\label{eq:k1}
\end{equation}
It follows that \(M+j\mid L_j\) and \(M+j\nmid L_{j+1}\). Similarly, by writing \[K^u=K^j(r,\partial_r,M+2u-2j)K^{u-j}\] it is apparent that
\begin{equation}
K^u(r,\partial_r,j-2u)=K^j_{0,0}(-j,H) K^{u-j}(r,\partial_r,j-2u)\label{eq:k2}
\end{equation}
so \(M+2u-j\mid L_j\) and \(M+2u-j\nmid L_{j+1}\). We know that \(M+u\mid L_u\). Thus,
\[\prod^{2u-j}_{i=j}(M+i)\]
divides \(L_j\) for all \(j=1,\ldots,|u|\). By \eqref{eq:kopcoeff}, we know that equality holds for \(j=1\), and so it must hold for all \(j\).  
\end{proof}
\begin{corollary}
Let \(u,v\) be integers with \(|u|\geq |v|\). Then the operator \[\bm{(}K^u-K^{|u|-|v|}_{0,0}(|v|-|u|,H)K^v\bm{)}(M+u+v)^{-1}\]
is polynomial in \(M\).\label{kcor}
\end{corollary}
We shall use these calculations to begin deriving the linear combinations which are admissible for an integral of \(H\). Since the coefficients of the \(J\)-operators have no common factor, it is only necessary to first consider \(K\)-operators by themselves and then extend to their products.
\begin{theorem}
Let \(u\) be a positive integer greater than one and let \(a,b_\pm,c_\pm,\cdots,z_\pm\) be a series of \(2j+1\) rational functions of \(M,H\) such that
\[Q=a+K^ub_++K^{-u}b_-+K^{2u}c_++K^{-2u}c_-+\cdots+K^{ju}z_++K^{-ju}z_-\]
is an operator polynomial in \(M,\partial_r\). Then \(z_\pm \prod^{2j-1}_{i=0}(M\pm iu)\) is a polynomial function of \(M,H\).\label{kred}
\end{theorem}
\begin{proof}
By denominator, we refer to the minimal polynomial \(\zeta_\pm\) in \(M\) such that \(z_\pm\zeta_\pm\) is also polynomial. Up to leading order, we have
\begin{align*}
Q&=(-2)^{ju-1}\Bigg[\left(\frac{M^{-1}c}{s^{2ju-1}}\partial_r-\frac{1}{s^{2ju}}\right)\prod^{2ju-1}_{i=0}(M+i)z_+\\
&\qquad-\left(\frac{M^{-1}c}{s^{2ju-1}}\partial_r+\frac{1}{s^{2ju}}\right)\prod^{2ju-1}_{i=0}(M-i)z_-\Bigg]+\mcl{O}\left(\frac{1}{s^{2ju-2}}\right)
\end{align*}
so
\[z_+\prod^{2ju-1}_{i=1}(M+i)-z_-\prod^{2ju-1}_{i=1}(M-i)\]
and
\[z_+\prod^{2ju-1}_{i=0}(M+i)+z_-\prod^{2ju-1}_{i=0}(M-i)\]
must be polynomials in \(M\).  The denominators of \(z_\pm\) therefore must divide \(\prod^{2ju-1}_{i=0}(M\pm i)\) respectively. If \(j=1\) then we are done, otherwise, we shall proceed inductively on \(j>1\). The common factor here between the two denominators is \(M\), so we deduce the existence of polynomial operators
\[Q'=a'+K^ub'_++K^{-u}b'_-+\cdots+K^{(j-1)u}y_+'+K^{-(j-1)u}y_-'+K^{ju}z_+'\]
and
\[Q''=a''+K^ub''_++K^{-u}b''_-+\cdots+K^{(j-1)u}y_+''+K^{-(j-1)u}y_-''+K^{-ju}z_-''\]
such that \(Q=(Q'+Q'')M^{-1}\). That the leading order terms of \(Q'\) may be polynomial in \(M\) and \(\partial_r\), the denominator of \(z'_+\) must divide \(\prod^{(2j-1)u}_{i=u}(M+i)\), by Lemma \ref{lemma}.

Let \(Q'''=Q'(M+u)[M+(2j-1)u]\) and denote by \((\cdot)'''\) all the coefficients therein. This ensures \[K^{ju}_{0,2(j-1)u} z'''_+\] and \[(\!K^{ju}_{1,2(j-1)u-1}z_+'''\] are polynomials. Then the denominators of \(y_\pm'''\) must divide \(\prod^{2(j-1)u-1}_{i=0}(M\pm i)\). In particular, the denominators of \(y_-'''\) and \(z_+'''\) share no common factor. We may thus resolve \(Q'''\) into a sum of two polynomial operators
\[Q^{\textrm{iv}}=a^{\textrm{iv}}+K^ub^{\textrm{iv}}_++K^{-u}b^{\textrm{iv}}_-+\cdots+K^{(j-1)u}y_+^{\textrm{iv}}+K^{ju}z_+^\textrm{iv}\]
and
\[Q^{\textrm{v}}=a^{\textrm{v}}+K^ub^{\textrm{v}}_++K^{-u}b^{\textrm{v}}_-+\cdots+K^{(j-1)u}y_+^{\textrm{v}}+K^{-(j-1)u}y_-^{\textrm{v}}.\]
Let \(Q^{\textrm{vi}}\) be the result of substituting \(M-2u\) for \(M\) in \(Q^{\textrm{iv}}\) and denote by \((\cdot)^\textrm{vi}\) all the coefficients therein. Then
\begin{align*}
Q^{\textrm{vi}}K^{-u}&=b^{\textrm{vi}}_+\Psi^{-u}(M,H)+K^uc_+^{\textrm{vi}}\Psi^{-u}(M,H)+K^{-u}a^{\textrm{vi}}\\
&\qquad +\cdots+K^{(j-1)u}z_+^{\textrm{vi}}\Psi^{-u}(M,H)+K^{-(j-1)u}x_-^{\textrm{vi}}
\end{align*}
is a polynomial operator. The denominator of \(z_+^{\textrm{vi}}\Psi^{-u}(M,H)\) thus divides
\[\prod^{2j-3}_{i=0}(M+iu).\]
There are no common factors between this product and \(\Phi^{-u}(M,H)\). So the denominator of \(z_+^{\textrm{iv}}\) divides
\[\prod^{2j-1}_{i=2}(M+iu).\]
But the denominator of \(z_+'''=z_+^{\textrm{iv}}\), by definition of \(Q'''\), cannot contain the factor \(M+(2j-1)u\). As \(z_+=z_+'''M^{-1}(M+u)^{-1}[M+(2j-1)u]^{-1}\), the result is proven.
\end{proof}

\section{Symmetry Algebra}
%
The construction of the symmetry algebra for superintegrable systems allows us to describe the degeneracies of the Hamiltonian. Each degenerate family consolidates into the whole eigenspace of the Hamiltonian, and so gives us the whole spectrum. 

For \(J^uK^{ku}\) to be well-defined, then \(u\) and \(ku\) must both be integers, so \(u\) is an integral multiple of \(q\). Therefore the general form of an integral of \(H\) is
\begin{equation}
I=\sum_{i\in\mbb{Z}}J^{iq}K^{ip}\chi_i\label{eq:Iform}
\end{equation}
where \(\chi_i\) are formal rational functions, with only finitely many that are non-zero, of \(M,H\) such that \(I\) is a polynomial in those operators, and, for the TTW model, even in \(M\). Let \(\chi_N^2+\chi_{-N}^2\neq 0\) and \(\chi_i=0\) for \(|i|>N\). Then by Theorem \ref{kred}, we require
\begin{equation}\frac{\chi_{\pm N}\prod^{2N-1}_{i=0}(M\pm ip)}{\gcd\{J^{\pm Nq}_{0,j}(M),J^{\pm Nq}_{1,j}(M)\}_{j\geq 0}}\label{eq:chi}\end{equation}
to be a polynomial in \(M\). By Theorems \ref{jpvzred} and \ref{jttwred}, the denominator in \eqref{eq:chi} is a unit. So, by linear independence of \(\{J^{iq}K^{ip}\}_i\), the integral must be of order \(\geq 2N(p+q-1)\) for the TTW model and \(\geq N(2p+q-2)\) for the PVZ model. Whether or not we can eliminate any of the factors in the product must depend, by Corollary \ref{kcor}, on

So far the representation theory for the TTW and the PVZ model has not been developed. In this section we will also provide a description of the representations from the symmetry algebra in the context of deformed oscillator realizations. This is also interesting in its own right as the construction of representations has been mainly limited to quadratic algebras such as Racah or some cubic algebras as in the Heun-Lie algebras.

The deformed oscillator realization \cite{dasko} was introduced within the context of quadratically superintegrable systems. Here we will rely on the analogous construction for polynomial algebras of three generators of arbitrarily large order. This realization consists in the system generated by the three operators \(\mathfrak{N},\mathfrak{b},\mathfrak{b}^\dagger\) satisfying the relations
\begin{align*}
[\mathfrak{b},\mathfrak{N}]&=\mathfrak{b},& [\mathfrak{N},\mathfrak{b}^\dagger]&=\mathfrak{b}^\dagger,& \mathfrak{b}^\dagger\mathfrak{b}&=\Xi(\mathfrak{N}),&\mathfrak{b}\mathfrak{b}^\dagger&=\Xi(\mathfrak{N})\end{align*}
where \(\Xi\) is, for our purposes, an analytic function. The finite-dimensional irreducible representations of this algebra are, up to a translation of \(\mathfrak{N}\), given by
\begin{align*}
\mathfrak{N}|i\rangle &=i|i\rangle,& \mathfrak{b}|i\rangle &=\sqrt{\Xi(i)}|i-1\rangle,&\mathfrak{b}^\dagger |i\rangle &=\sqrt{\Xi(i+1)}|i+1\rangle ,
\end{align*}
where \(\{|i\rangle\}^{n}_{i=0}\) is the basis of the module. The constraints on finite dimension imply that \(\Xi(0)=\Xi(n)=0\) with no intermediate roots of \(\Xi\). The central elements, namely \(\mbb{R}[H]\) act as constants on the representation, and we denote the eigenvalue of \(H\) as \(E\). with no other roots of \(\Xi\). From one of the two roots of \(\Xi\), we may determine the eigenvalue \(E\) of \(H\). The energy spectrum will include as a subset the physical spectrum \eqref{eq:spectrum} found by analysis. However, even the non-physical states can still be interesting in its links with finding new special functions.

\subsection{PVZ q even}\label{section1}
For \(q\) even, \(2p+q\geq N(2p+q-2)\) for all \(N>1\). Thus we need only consider the form
\begin{align*}
I&=\phi_0(M,H)M^{-1}(M^2-p^2)^{-1}+J^qK^p\phi_1(M,H)M^{-1}(M+p)^{-1}\\
&\qquad+J^{-q}K^{-p}\phi_{-1}(M,H)M^{-1}(M-p)^{-1}
\end{align*}
to find the smallest integral with \(\phi_0\) a polynomial and \(\phi_{\pm 1}\) at most quadratic. Let us show that an integral of order \(2p+q-2\) or \(2p+q-1\) is impossible. We note that
\[\Res_{M=0}I=k^q(2q-1)!![\mathrm{e}^{-2p\theta R}\phi_1(0,H)-\mathrm{e}^{2p\theta R}\phi_{-1}(0,H)]K^p(r,\partial_r,0)/p+\cdots\]
where the ellipses indicate lower frequency terms, cannot be zero unless \(M\) divides \(\phi_{\pm 1}\). Similarly,
\[\Res_{M=\pm p}I=\pm k^q\mrm{e}^{\pm 2p\theta R}\phi_{\mp 1}(\pm p,H)K^p_{0,0}(-p,H)\prod^q_{i=1}(2i-1-q)/p+\cdots\]
requires that \(\Phi_{\mp 1}(\pm p,H)=0\). Therefore, \(H\) is a superintegrable Hamiltonian of order \(2p+q\) with generators \(M,J^qK^p,J^{-q}K^p\). We set
\begin{align*}
A&=\frac{M}{2p},& 
B&=J^qK^p,& C=J^{-q}K^{-p},
\end{align*}
and \(\Lambda_\pm=\Phi^{\pm q}(2pA)\Psi^{\pm p}(2pA,H)\). The symmetry algebra relations are
\begin{subequations}
\begin{align}
\label{eq:ab}[A,B]&=B,\\ \label{eq:ac}[A,C]&=-C,\\
\label{eq:bc}[B,C]&=\Lambda_--\Lambda_+.
\end{align}
\end{subequations}
The algebra of order \(4p+2q-1\) if \(\kappa\neq 0\) and of order \(2p+2q-1\) if \(\kappa=0\). The enveloping algebra of \(A,B,C,H\), regarded as purely abstract operators, joined by the relations \eqref{eq:ab}, \eqref{eq:ac}, \eqref{eq:bc} with \(H\) central, satisfies the PBW property. Furthermore, the operator
\[D=\{B,C\}-\Lambda_+-\Lambda_-\]
commutes with \(A,B,C\). The differential operator realization annihilates \(D\) so that the centre of the algebra is \(\mbb{R}[H]\). 

The realisation as a deformed oscillator algebra is straightforward:
\begin{align*}
A&=\mfk{N}+u,\\
B&=\mathfrak{b}^\dagger+\mathfrak{b},\\
C&=\mathfrak{b}^\dagger-\mathfrak{b},\\
\Xi(\mathfrak{N})&=\Phi^{-q}(2p[\mathfrak{N}+u])\Psi^p(-2p[\mathfrak{N}+u],E),\\
\Xi(\mathfrak{N}+1)&=\Phi^q(2p[\mathfrak{N}+u])\Psi^p(2p[\mathfrak{N}+u],E).
\end{align*}
where \(u\) is an arbitrary constant. 
\subsection{PVZ q odd}\label{section2}
For \(q\) odd, \(2p+q-1\geq N(2p+q-2)\) for all \(N>1\). So as before, we consider
\begin{align*}
I&=\phi_0(M,H)M^{-1}(M^2-p^2)^{-1}+J^qK^p\phi_1(M,H)M^{-1}(M+p)^{-1}\\
&\qquad+J^{-q}K^{-p}\phi_{-1}(M,H)M^{-1}(M-p)^{-1}.
\end{align*}
By a similar argument to \S 5.1, \(\phi_0,\phi_{\pm 1}\) must divide \(M\). Lemma 4.1 guarantees that may take \(\phi_{\pm 1}\) to be any real number proportion of \(M\). Therefore, \(H\) is a superintegrable Hamiltonian of order \(2p+q-1\). If we write
\begin{align*}
X&=[J^qK^p-J^q_{0,0}(-p)K^p_{0,0}(-p,H)](M+p)^{-1},\\
Y&=[J^{-q}K^{-p}-J^{-q}_{0,0}(p)K^p_{0,0}(-p,H)](M-p)^{-1},
\end{align*}
then \(X,Y,M,H\) together generate all the integrals of \(H\). 
The Hamiltonian is therefore super-integrable of order \(2p+q-1\). We take
\begin{align*}
A&=\frac{M}{2p}; & B&=p(Y+X); & C&=p(Y-X)
\end{align*}
and
\[\Lambda_\pm=\frac{\Phi^{\pm q}(2pA)\Psi^p(\pm 2pA,H)-[J^{\pm q}_{0,0}(\mp p)K^p_{0,0}(-p,H)]^2}{(2A\pm 1)^2}.\]
The algebra relations are
\begin{align*}
\{A,B\}&=C-[J^q_{0,0}(-p)+J^{-q}_{0,0}(p)]K^p_{0,0}(-p,H),\\
\{A,C\}&=B+[J^q_{0,0}(-p)-J^{-q}_{0,0}(p)]K^p_{0,0}(-p,H),\\
\{B,C\}&=2\Lambda_+-2\Lambda_-.
\end{align*}
This is a symmetry algebra of order \(4p+2q-3\) if \(\kappa\neq 0\) and of order \(2p+2q-3\) if \(\kappa=0\). We have the Casimir operator of the enveloping algebra
\begin{align*}
D&=B^2+C^2-2\Lambda_+-2\Lambda_-.
\end{align*}
The element \(D\) evaluates to zero in the differential-difference operator realisation.

When \(k=1\), we have
    \begin{align*}
    \{B,C\}&=-2\alpha\beta(4\omega^2-\kappa^2+4\kappa H)+4A[4\omega^2+4\kappa H+\kappa^2(\alpha^2+\beta^2-2)]\\
    &\qquad+8\kappa^2\alpha\beta A^2-32\kappa^2A^3.
    \end{align*}
For \(\kappa=0\), \(\Lambda_+-\Lambda_-\) is of first-order, and the system generated by \(A,B,C\) is a Bannai-Ito algebra over the ring \(\mbb{R}[H]\):
\begin{align*}
\{A,B\}&=C-2\alpha H;\\
\{A,C\}&=B+2\beta H;\\
\{B,C\}&=2(4\omega^2-\kappa^2)(2A-\alpha\beta);\\
0&=(4\omega^2-\kappa^2)(4A^2+1-\alpha^2-\beta^2)+B^2+C^2-4H^2.
\end{align*}
In order to produce a deformed oscillator realization, we require a `fermionic' number operator in contrast to the `bosonic' operator \(\mathfrak{N}\) that alternates the sign of \(\mathfrak{b},\mathfrak{b}^\dagger\) for left- and right-multiplication. This is achieved if we take \((-1)^{\mathfrak{N}}=\cos(\pi\mathfrak{N})\), giving:
\begin{align*}
A&=-(\mathfrak{N}+u)\cos(\pi\mathfrak{N}),\\
B&=\mathfrak{b}[2(\mathfrak{N}+u)-1]^{-1}+[2(\mathfrak{N}+u)-1]^{-1}\mathfrak{b}^\dagger\\
&\qquad+K^p_{0,0}(-p,E)\left[\frac{J^q_{0,0}(-p)}{2(\mathfrak{N}+u)\cos(\pi\mathfrak{N})-1}-\frac{J^{-q}_{0,0}(p)}{2(\mathfrak{N}+u)\cos(\pi\mathfrak{N})+1}\right]\\
C&=\cos(\pi\mathfrak{N})\{\mathfrak{b}[2(\mathfrak{N}+u)-1]^{-1}-[2(\mathfrak{N}+u)-1]^{-1}\mathfrak{b}^\dagger\}\\
&\qquad-K^p_{0,0}(-p,E)\left[\frac{J^q_{0,0}(-p)}{2(\mathfrak{N}+u)\cos(\pi\mathfrak{N})-1}+\frac{J^{-q}_{0,0}(p)}{2(\mathfrak{N}+u)\cos(\pi\mathfrak{N})+1}\right],\\
\Xi(\mathfrak{N})&=\Phi^{-q\cos(\pi\mathfrak{N})}(-2p[\mathfrak{N}+u]\cos[\pi\mathfrak{N}])\Psi^p(-2p[\mathfrak{N}+u],E),\\
\Xi(\mathfrak{N}+1)&=\Phi^{q\cos(\pi\mathfrak{N})}(-2p[\mathfrak{N}+u]\cos[\pi\mathfrak{N}])\Psi^p(2p[\mathfrak{N}+u],E).
\end{align*}
\subsection{TTW on constant-curvature space}
For the TTW model, the operators
\begin{align*}
A&=\frac{M^2}{4p^2}-\frac{1}{4};\\
B&=p^2J^qK^pM^{-1}(M+p)^{-1}+p^2J^{-q}K^{-p}M^{-1}(M-p)^{-1}\\
&\qquad-2p^2J^q_{0,0}(-p)K^p_{0,0}(-p,H) (M^2-p^2)^{-1}\\
C&=pJ^qK^pM^{-1}-pJ^{-q}K^{-p}M^{-1};
\end{align*}
together with \(H\), generate all the integrals of this model and thus the full symmetry algebra. So \(H\) is superintegrable of order \(2p+2q-2\) (this being the order of \(B\)). The symmetry algebra relations are
\begin{align}
[A,B]&=C;\\
[A,C]&=2\{A,B\}+2J^q_{0,0}(-p)K^p_{0,0}(-p,H); \label{eq:[a,c]} \\
[B,C]&=-2B^2+\Lambda \label{eq:[b,c]}
\end{align}
where
\begin{equation}
\Lambda=\tfrac{2p^3[(M-p)^2\Phi^q(M)\Psi^p(M,H)-(M+p)^2\Phi^q(-M)\Psi^p(-M,H)+4pM\langle q;-p;0\rangle^2\langle\!\langle p;-p,H;0\rangle\!\rangle^2]}{M(M^2-p^2)^2}.
\end{equation}
is a polynomial in \(A,H\). This algebra is of order: \(2p+2q-2\) if \(\kappa\neq 0\); \(p+2q-2\) if \(\kappa =0,k\neq 1\); or \(2\) if \(\kappa=0,k=1\). There is a Casimir
\begin{equation}
D=C^2+4B^2-2\{A,B^2\}-4J^q_{0,0}(-p)K^p_{0,0}(-p,H) B+\Omega\label{eq:omega}
\end{equation}
where
\begin{equation}
\Omega=\tfrac{2p^2[(M-p)^3\Phi^q(M)\Psi^p(M,H)+(M+p)^3\Phi^q(-M)\Psi^p(-M,H)-2\Phi^q(-p,H)\Psi^p(-p,H)M(M^2+3p^2)]}{M(M^2-p^2)^2}
\end{equation}
is a polynomial in \(A,H\). In the differential operator realisation, \(D\equiv 0\). 

For \(k=1\), these polynomials evaluate to
\begin{align*}
\Lambda&=8H^2+16\sigma(\alpha^2+\beta^2-2-2A)+16\kappa H(\alpha^2+\beta^2-2-2A)\\
&\qquad+2\kappa^2[(\alpha^2-\beta^2)^2+12A^2-2(\alpha^2+\beta^2)(4A+3)+28A+12];\\
\Omega&=8 (\alpha^2 + \beta^2) (4\sigma A - H^2) - 
 4 (\alpha^2 - \beta^2)^2\sigma - 64\sigma A^2\\  
 &\qquad+4\kappa H[
   8 (\alpha^2 + \beta^2 - 2) A - (\alpha^2 - \beta^2)^2 - 16 A^2] \\&\qquad+ 
 4\kappa^2 A[
   16 A^2 + (\alpha^2 - \beta^2)^2 (A - 1) - 44 A + 48 - 
    2 (\alpha^2 + \beta^2) (4 A + 3)].
\end{align*}
If, in addition, \(\kappa=0\), then \(\Lambda\) is of first-order in \(A\) and we obtain a Racah algebra over the ring \(\mbb{R}[H]\):
\begin{align*}
[A,B]&=C,\\
[A,C]&=2\{A,B\}+2H(\alpha^2-\beta^2),\\
[B,C]&=-2B^2+8H^2+16\sigma (\alpha^2+\beta^2-2-2A),\\
0&=C^2+4\beta^2-2\{A,B^2\}-4H(\alpha^2-\beta^2)B\\&\qquad+8(\alpha^2+\beta^2)(4\sigma A-H^2)-4(\alpha^2-\beta^2)^2\sigma-64\sigma A^2.
\end{align*}
A deformed oscillator realization is readily given by
\begin{align*}
A&=(\mathfrak{N}+u)^2-\tfrac{1}{4};\\
B&=\mathfrak{b}^\dagger(\mathfrak{N}+u)^{-1}[2(\mfk{N}+u)+1]^{-1}+[2(\mfk{N}+u)+1]^{-1}(\mathfrak{N}+u)^{-1}\mfk{b}\\
&\qquad-2J^q_{0,0}(-p)K^p_{0,0}(-p,E) [4(\mfk{N}+u)^2-1]^{-1},\\
C&=\mfk{b}^\dagger(\mathfrak{N}+u)^{-1}-(\mathfrak{N}+u)^{-1}\mfk{b},\\
\Xi(\mathfrak{N})&=4\Phi^q(-2p[\mfk{N}+u])\Psi^p(-2p[\mfk{N}+u],E),\\
\Xi(\mathfrak{N}+1)&=4\Phi^q(2p[\mathfrak{N}+u])\Psi^p(2p[\mfk{N}+u],E).
\end{align*}
\subsection{Spectrum}
The deformed oscillator realizations of both models relate the roots of \(\Xi\) to those of \(\Phi^{\pm q},\Psi^p\). The roots of these structure functions supply constraints on \(n,u\) and allow us to determine the value of \(E\) corresponding to a \(n+1\)-dimensional representations.

Let us consider the case where \(\Phi^{-q}(2pu)=\Phi^q(2p[n+u])=0\) for the TTW model and the PVZ model with \(q\) even and \(\Phi^{-q}(2pu)=\Phi^{(-1)^nq}(2p[-1]^n[n+u])=0\) otherwise. Then \(E\) is free to take any value but at least one of \(\alpha,\beta,\alpha+\beta,\alpha-\beta\) must be an integer (an even integer, if PVZ). Second, we have \(\Psi^p(-2pu,E)=\Psi^p(2p[n+u],E)=0\). This yields the two equations:
\begin{align}
    0&=E^2-(\omega^2+\kappa E)(2pu-2i-1)^2+\tfrac{1}{4}\kappa^2(2pu-2i-1)^4,\label{eq:psi1}\\
    0&=E^2-(\omega^2+\kappa E)(2pn+2pu+2i'+1)^2+\tfrac{1}{4}\kappa^2(2pn+2pu+2i'+1)^4,\label{eq:psi2}
\end{align}
for some \(0\leq i,i'<p\). Eliminating \(E\) results in an equation of fourth degree in \(u\). There is first, the double root:
\[u_0=\frac{i-i'-pn}{2p}.\]
Setting \(u=u_0\) makes \eqref{eq:psi1} identical to \eqref{eq:psi2}. There are two corresponding values for \(E\), given by
\[E=\tfrac{1}{2}\kappa(pn+i+i'+1)^2\pm \omega (pn+i+i'+1).\]
There is another value eigenvalue:
\[E=\tfrac{1}{2}\kappa(pn+i+i'+1)^2-\frac{\omega^2}{2\kappa}\]
which only exists for non-zero curvature. This occurs when \(u\) satisfies the two other roots of the resultant:
\[u_\pm=u_0\pm\frac{\omega}{2p}.\]
Finally, the case where \(\Phi^{-q}(2pu)=\Psi^q(2p[n+u])=0\) or \(\Psi^p(-2pu,E)=0\) and \(\Phi^q(2p[n+u])=0\) for the TTW model and the PVZ model with \(q\) even and \(\Phi^{(-1)^nq}(2p[-1]^n[n+u])=0\) otherwise. This yields eight different sequences of eigenvalues:
\[E_{m,\ell}=\pm\omega \epsilon_{m,\ell}+\tfrac{1}{2}\kappa \epsilon^2_{m,\ell},\qquad \epsilon_{m,\ell}=2m+1+k(2\ell+\gamma+1)\]
where \(\gamma\in \{\alpha+\beta,\alpha-\beta,-\alpha+\beta,-\alpha-\beta\}\). The physical energies correspond to a choice of positive sign for \(\alpha,\beta,\omega\).
\section{Conclusion}
We have applied the expansion method to the TTW and PVZ model. This has proved successful on account of their separability, where we are able to obtain ODE in the radial variable with operator coefficients. The solution lead us to ladder operators connected with orthogonal polynomials which we computed and analyzed their properties in order to obtain the respective symmetry algebras. This shows how the expansion method can lead to the resolution of a range of superintegrable systems. The approach here has also proved more effective in analyzing the ladder operators as we have not needed to consider the action on a basis of eigenfunctions. Moreover, we have shown that the ladder operators do not give the generators directly but have given a method to reduce them down to basic operators. 

Other models in two dimensions, such as the Post-Winternitz systems, could be considered as well as various Darboux deformations which provide isospectral or almost isospectral problems, which lead to exceptional Jacobi polynomial of Type I, II and III \cite{oda11,ull14}. In addition, we may also consider the case where the angular part consists of exotic non-elementary functions, since explicit knowledge of the eigenfunctions is not readily obtainable \cite{tre10,rita,ritb}. 

We have also found the degenerate energy spectrum of the two Hamiltonians via the algebraic method. That is, the three generator polynomial algebra formed by the integrals and the deformed oscillator realization. It is not known in general if the spectrum can always be recovered in this way for superintegrable systems, even on two-dimensional spaces. This paper has given examples where this is the case for a symmetry algebra which may be of arbitrarily large order.

\section*{Acknowledgements}
The research of Ian Marquette was supported by Australian Research Council Future Fellowship FT180100099.

\printbibliography[title={References}]

\end{document}